\documentclass[pdflatex,sn-mathphys-num]{sn-jnl}

\usepackage{graphicx}%
\usepackage{multirow}%
\usepackage{amsmath,amssymb,amsfonts}%
\usepackage{amsthm}%
\usepackage{mathrsfs}%
\usepackage[title]{appendix}%
\usepackage{xcolor}%
\usepackage{textcomp}%
\usepackage{manyfoot}%
\usepackage{booktabs}%
\usepackage{algorithm}%
\usepackage{algorithmicx}%
\usepackage{algpseudocode}%
\usepackage{listings}%

\theoremstyle{thmstyleone}%
\newtheorem{theorem}{Theorem}
\theoremstyle{thmstyletwo}%
\newtheorem{example}{Example}%

\theoremstyle{thmstylethree}%
\newtheorem{definition}[theorem]{Definition}%
\newtheorem{assumption}[theorem]{Assumption}
\newtheorem{lemma}[theorem]{Lemma}

\newcommand\F{\mbox{I\kern-2pt F}}
\newcommand\cA{{\cal A}}
\newcommand\cE{{\cal E}}
\newcommand\cC{{\cal C}}
\newcommand\cF{{\cal F}}
\newcommand\cG{{\cal G}}
\newcommand\cH{{\cal H}}
\newcommand\cI{{\cal I}}

\newcommand\cL{{\cal L}}
\newcommand\cB{{\cal B}}
\newcommand\cN{{\cal N}}
\newcommand\cM{{\cal M}}
\newcommand\cX{{\cal X}}
\newcommand\cD{{\cal D}}

\newcommand\cP{{\cal P}}

\newcommand\cT{{\cal T}}

\newcommand\cV{{\cal V}}

\def\bbr{{\mathbb R}}
\def\bbn{{\mathbb N}}

\def\bbf{{\mathbb F}}
\def\bbg{{\mathbb G}}

\def\P{{\bf P}}
\newcommand\E{{\bf E}}
\newcommand\1{{\bf 1}}
\newcommand{\wh}{\widehat}

\newcommand\beq{\begin{equation}}
\newcommand\eeq{\end{equation}}
\newcommand\bea{\begin{eqnarray}}
\newcommand\eea{\end{eqnarray}}
\newcommand\bean{\begin{eqnarray*}}
\newcommand\eean{\end{eqnarray*}}

\begin{document}

\title{Well-posedness of behavioral singular stochastic control problems}


\author*[1,2,3]{\fnm{Artur} \sur{Sidorenko}}\email{sidorenkoap@my.msu.ru}

\affil*[1]{\orgname{Lomonosov Moscow State University}, \orgaddress{ \city{Moscow}, \country{Russia}}}

\affil[2]{\orgname{Vega Institute Foundation}, \orgaddress{ \city{Moscow}, \country{Russia}}}

\affil[3]{\orgname{Higher School of Economics}, \orgaddress{ \city{Moscow}, \country{Russia}}}





\abstract{We investigate the well-posedness of a general class of singular stochastic control problems in which controls are processes of finite variation. We develop an abstract framework, which we then apply to storage management and portfolio investment problems under proportional transaction costs. Within this setting, we establish the existence of an optimal strategy in the class of randomized controls for a range of goal functionals, including cumulative prospect theory (CPT) preferences. To the best of our knowledge, this is the first treatment of behavioral storage management with the CPT goal functional. For the portfolio management problem, our analysis exploits the metrizable Meyer–Zheng topology to simplify proofs. We thoroughly investigate the applicability of the Skorokhod representation theorem for adapted random processes. Overall, the presented framework provides a clear separation between constraints imposed on the market model and properties required of the goal functional.
}

\keywords{Singular stochastic control, Skorokhod representation theorem, Markets with transaction cost, Inventory management, Cumulative prospect theory, Meyer--Zheng topology}

\pacs[MSC Classification]{49J55, 60B10, 90B05, 91G10, 93E20}

\pacs[JEL Classification]{C65}

\maketitle

\section{Introduction}

Singular stochastic control problems arise in diverse contexts, including operations research, finance, and inventory theory. Early work was motivated by engineering applications  \citep{bather1967sequential} and storage management models \citep{harrison1983instantaneous,karatzas1983singular}, see also \citep{dixit1991simplified,sulem1986solvable}. A mixed singular and impulse control problem was considered in \citep{dai2013brownian,dai2013brownian2}. A comprehensive survey of singular control in inventory theory can be found in \citep{perera2023survey}. Most of the research models the demand in the storage dynamics as Brownian motion or a diffusion process. However, recent research involves regime-switching  \citep{shao2023optimal} and L\'evy processes \citep{yamazaki2017inventory,noba2025refraction} as well as control with unknown model parameters \citep{federico2023two}. 

Beyond inventory models, singular control has applications in mathematical finance, particularly in markets with transaction costs, where asset positions are adjusted via processes of finite variation. Three distinct paths have shaped the research landscape: Markovian approach \citep{magillPortfolioSelectionTransactions1976a,davisPortfolioSelectionTransaction1990a,shreveOptimalInvestmentConsumption1994a},  martingale approach \citep{cvitanicHEDGINGPORTFOLIOOPTIMIZATION1996,cvitanicOptimalTerminalWealth2001,guasoniOptimalInvestmentTransaction2002a,guasoniNOARBITRAGETRANSACTION2006a,czichowskyDualityTheoryPortfolio2016}, and weak convergence approach \citep{bankScalingLimitSuperreplication2017,bayraktarExtendedWeakConvergence2020,chauBehavioralInvestorsConic2020,bayraktarShortCommunicationNote2021}. 
Unlike problems for frictionless markets, price processes are not necessarily semimartingales \citep{guasoniOptimalInvestmentTransaction2002a,guasoniNOARBITRAGETRANSACTION2006a}.  

A convenient geometric formulation of singular control problems in finance was suggested in \citep{kabanovHedgingLiquidationTransaction1999,kabanovHedgingTransactionCosts2002}, see also \citep{kabanovMarketsTransactionCosts2010} for a comprehensive survey.  This setting was investigated in \citep{deelstraDualFormulationUtility2001,bouchardUtilityMaximizationReal2002} by duality arguments, whereas \citep{kabanovGeometricApproachPortfolio2004a,devalliereConsumptioninvestmentProblemTransaction2016} treated the HJB equations. Later, \citep{campiSuperreplicationTheoremKabanov2006,campiMultivariateUtilityMaximization2011} incorporated discontinuous price processes. Other works include \citep{kabanovHedgingLiquidationTransaction1999,kabanovHedgingTransactionCosts2002,schachermayerFundamentalTheoremAsset2004,grepatMultiassetVersionKusuoka2021}. 

Traditionally, these problems are formulated in terms of expected gain or utility maximization. However, decision-makers often depart from expected utility maximization as illustrated by the Allais paradox. To address this issue, cumulative prospect theory (CPT) was developed \citep{tverskyAdvancesProspectTheory1992,tversky1974judgment}.
Optimal control under CPT preferences has been studied primarily in frictionless markets \citep{jinBEHAVIORALPORTFOLIOSELECTION2008}, see also \citep{hePORTFOLIOCHOICEQUANTILES2011,xuNOTEQUANTILEFORMULATION2016,heHOPEFEARASPIRATIONS2016,ruschendorfConstructionOptimalPayoffs2020,vanbilsenDynamicConsumptionPortfolio2020,daiNonconcaveUtilityMaximization2022,biOptimalInvestmentProblem2023}. Well-posedness of such problems was examined \citep{rasonyiOptimalPortfolioChoice2013}.

In this study, we develop a general optimization framework that naturally embeds a range of applicable singular control problems under various goal functionals, including expected utility maximization and CPT preferences. The crucial features of this framework are: controls are processes of finite variation, a countable number of constraints is imposed, and the target functional depends solely on the probability distribution in a path space (i.e. expected utility or CPT preferences). Within this framework, we establish sufficient conditions for the existence of an optimal solution. 

Using this framework, we establish existence of an optimal control in a class of randomized strategies for both inventory management and portfolio investment problems involving transaction costs and CPT preferences. The use of randomized strategies is motivated by the observation that such strategies can improve the satisfaction of agents with CPT preferences \citep{carassusOPTIMALINVESTMENTBEHAVIORAL2015}. The additional randomness was exploited in other settings as well, see, e.g., \citep{hendersonRandomizedStrategiesProspect2017,hu2023casino}.  This effect bears some resemblance to the difference between weak and strong solutions of SDE, see, for example, \citep{tsirelsonExampleStochasticDifferential1976}. A key strength of our framework is its ability to clearly separate assumptions on the performance functional from those on the underlying stochastic model.  

To the best of our knowledge, this paper is the first treatment of behavioral storage management. Incorporating the CPT preferences into the framework would enable tailoring the solutions to the risk preferences of real economic agents. Our framework enables going beyond semimartingales and verifying well-posedness of the inventory problem when the demand is an arbitrary process with continuous paths.

For the portfolio investment problem, only \citep{chauBehavioralInvestorsConic2020} treats continuous-time markets with friction and CPT investors, whereas \citep{chauSkorohodRepresentationTheorem2017} considers non-linear frictions. These papers employ technically intricate weak convergence arguments for non-metrizable topologies. In singular control theory, the Skorokhod $J_1$ \citep{budhiraja2006existence} and weak $M_1$ \citep{cohen2021singular} topologies have been applied as well. 
We adopt a metrizable topology by P.A. Meyer and W.A. Zheng \citep{meyer1984tightness}, see also \citep{kurtzRandomTimeChanges1991}. Applications of this topology in portfolio optimizations can be found in \citep{bayraktarExtendedWeakConvergence2020}. The reliance on metrizable topologies simplifies technical proofs. 

The paper is organized as follows. Section  \ref{section:general} introduces the general setting. Section \ref{section:kabanov_setting} applies this setting for the multi-asset conic market and several examples of goal functionals including the CPT functional and discusses the embedding into the abstract setup. Finally, Section \ref{Section:storage} discusses the well-posedness of the behavioral storage management problem. For fluidity, proofs of the main results are relegated to Section \ref{Section:proof}, whereas auxiliary statements are collected in the Appendix.

\section{General setting}
\label{section:general}

Let $(\Omega, \cF, P)$ be a complete probability space, $Y$ be a $\bbr^m$-valued RCLL process with independent increments and $\xi$ be a uniform random variable on $[0, 1]$ independent of $Y$. Fix a strictly positive real number $T$. Let the filtration $\bbf = (\cF_t)_{t \in [0, T]}$ with $\cF_t = \sigma \{ \cN, \, Y_s, \, 0 \leq s \leq t \}$ be generated by $Y$ and the null sets $\cN$ from $\cF$. Standard arguments imply the right-continuity of $\bbf$. We define a filtration $\bbg = (\cG_t)_{t \in [0, T]}$ with $\cG_t = \sigma \{ \cF_t, \, \xi \}$. Again, $\bbg$ is right-continuous. 
We denote the space of continuous paths by $\cC_d := C([0, T], \bbr^d)$ and the Skorokhod space by $\cD_m := D([0, T], \bbr^m)$. 
\begin{definition}
    Any $\bbr^d$-valued $\bbg$-adapted c\`adl\`ag process with paths of finite variation is called a control or strategy.
\end{definition}
Denote the set of all controls by $\mathfrak{B}$.  

Let $\cV$ be a family of all right-continuous functions $f: \, [0, T] \to \bbr^d$ of finite variation ${\rm Var}_T \, f := \sum_i {\rm Var}_T \, f^i < \infty$, where ${\rm Var}_T \, f^i$ is the total variation of $f^i$ on $[0, T]$ with a convention $f^i(0-): =0$. To be more precise, 
$$
{\rm Var}_T \, f = \sup_{\tau \in \cT} \left[ |f(t_0)| + \sum_j |f(t_j) - f(t_{j-1})| \right],$$
where $\cT$ is the family of finite partitions $0 = t_0 < \dots < t_k = T$, $k \in \bbn$, of $[0, T]$.

We endow $\cV$ with the Meyer--Zheng topology \citep{meyer1984tightness} . Recall that this topology is initially defined on the the Skorokhod space $D_{MZ}(\bbr_+, \bbr^d)$. The subscript underpins the choice of the topology. We use a metric of \citep{bayraktarExtendedWeakConvergence2020} tailored to the space $\cD_{d, MZ} := D_{MZ}([0, T], \bbr^d)$ of functions on the finite segment $[0, T]$:
\beq
d_{MZ}(f, g) = \int_{[0, T[} \min \left(|f(t) - g(t)|, \, 1 \right) \, dt + \min\left(|f(T) - g(T)|, \, 1 \right). 
\eeq
The above Meyer--Zheng metric is a metric of convergence in measure ${\rm Leb}([0, T]) + \delta_T$. Note that any function $f \in \cD^d_{T, MZ}$ can be embedded into the space $D_{MZ}(\bbr_+, \bbr^d)$ by a continuous mapping $\Xi: \cD_{d, MZ} \to D_{MZ}(\bbr_+, \bbr^d)$ with $\Xi(f)(t) = f(t \wedge T)$. The space $D_{MZ}(\bbr_+, \bbr^d)$ is separable, though not complete, see \citep{meyer1984tightness}. Thus $\cD_{d, MZ}$ is separable. Recall that compactness is not required for the forward part of the Prokhrorov theorem.  The Borel $\sigma$-algebras generated by both topologies coincide: $\cB(\cD_{d, MZ}) = \cB(\cD_d) = \cB( \pi_t: \, t \in [0, T])$, where $\pi_t(f) = f(t)$ (see Lemma 4.2 of \citep{buckdahnWeakSolutionsBackward2005} and Theorem VI.1.14 of \citep{jacodLimitTheoremsStochastic2003} respectively). In order to establish tightness of probability measures under this topology, we use the Helly theorem: for any $c > 0$, the set $\{f \in \cD_{d, MZ}: \, {\rm Var}_T f \leq c \}$ is compact (see, e.g., Theorem 12.7 of \citep{protterFirstCourseReal1991} or Lemma 8 of \citep{meyer1984tightness}). In particular, $\cV$ is a Borel subset of $\cD^d_{T, MZ}$. A more detailed treatment is given in \citep{kurtzRandomTimeChanges1991} and \citep{buckdahnWeakSolutionsBackward2005}. 

For a separable metric space $\cE$, we introduce a Borel mapping $\phi: \cD_m \to \cE$ and an upper-semi-continuous function $\Phi: \cD_m \times \cE \times \cV \to \bbr^\bbn$.  Finally, we fix a goal functional $J: \cP(\cD_m \times \cE \times \cV) \to \bbr \cup \{ \pm \infty \}$.
\begin{definition}
    The control $B$ is admissible if $\Phi(Y, \phi(Y), B) \in \bbr_+^\bbn$ with probability one.
\end{definition}
Let $\cA$ be the set of all admissible controls.

The optimization problem is defined in the following way:
\beq\
\label{eq:sup}
\sup_{B \in \cA} J\left(\cL_P (Y, \phi(Y), B)\right),
\eeq
where $\cL_P(\xi)$ is the law of $\xi$ under measure $P$. 

For brevity, we put $\mathfrak{M} := \{ \cL_P (Y, \phi(Y), B ): \: B \in \cA \}$.

To guarantee attainability of the optimization problem, we impose three crucial assumptions.

\begin{assumption}
\label{A:1}
The function $J$ is upper-semi-continuous in $\mathfrak{M}$, i.e. if $\mu_n \overset{w}{\to} \mu$ in $\mathfrak{M}$ then 
\begin{equation*}
\limsup_n J(\mu_n) \leq J(\mu).
\end{equation*}
\end{assumption}

\begin{assumption}
\label{A:2}
The function $\Phi = (\Phi_i)_{i \in \bbn}$ is upper-semi-continuous, i.e. for every entry $i$, and every sequence $x_n \to x$ in $\cD_m \times \cE \times \cV $, we have 
\begin{equation*}
\limsup_{n \to \infty} \Phi_i(x_n) \leq \Phi_i(x).
\end{equation*}
\end{assumption}

\begin{assumption}
\label{A:3}
The set of random variables $\{ {\rm Var}_T B \, : B \in \cA \}$ is tight.    
\end{assumption}

The chief result of this section is as follows. 

\begin{theorem}
\label{theo:general}
Let Assumptions \ref{A:2} and \ref{A:3} be satisfied. Then the set $\mathfrak{M}$ is (sequentially) compact in $\cP(\cD_m \times \cE \times \cV)$.
\end{theorem}
The proof is inspired by \citep{chauSkorohodRepresentationTheorem2017} and is deferred to Section \ref{Section:proof}. The following theorem is the main result of this section. 
\begin{theorem} \label{theo:main}
Let $\cA \neq \varnothing$ and Assumptions \ref{A:1}--\ref{A:3} hold. Then the supremum of \eqref{eq:sup} is attainable, i.e. there exists a process $B^\dagger \in \cA$ such that 
\beq\
\label{eq:sup2}
\sup_{B \in \cA} J\left(\cL_P (Y, \phi(Y), B)\right) = J\left(\cL_P (Y, \phi(Y), B^\dagger)\right).
\eeq
\end{theorem}
\begin{proof}
If $\sup_{\nu \in \mathfrak{M}} J(\nu) = -\infty$, there is nothing to prove. Otherwise, we select such a sequence $\mu_n \in \mathfrak{M}$ satisfying 
\begin{equation*}
J(\mu_n) > \sup_{\nu \in \mathfrak{M}} J(\nu) - \frac{1}{n}
\end{equation*}
for a finite $\sup_{\nu \in \mathfrak{M}} J(\nu)$ or
\begin{equation*}
J(\mu_n) > n
\end{equation*}
for $\sup_{\nu \in \mathfrak{M}} J(\nu) = + \infty$. Due to compactness of $\mathfrak{M}$, there exists a subsequence of measures $\mu_{n_k} \to \mu^\dagger$ with the limit $\mu^\dagger \in \mathfrak{M}$. Due to Assumption \ref{A:1} and the choice of the sequence,
\begin{equation*}
\sup_{\nu \in \mathfrak{M}} J(\nu) \leq \limsup_{k \to \infty} J(\mu_{n_k}) \leq J(\mu^\dagger) \leq \sup_{\nu \in \mathfrak{M}} J(\nu).
\end{equation*}
Thus $J(\mu^\dagger) = \sup_{\nu \in \mathfrak{M}} J(\nu)$, and this concludes the proof.
\end{proof}

A detailed proof of Theorem \ref{theo:general} is relegated to Section \ref{section:proof_general}.


 \section{Portfolio investment under transaction costs}
 \label{section:kabanov_setting}

In this section, we describe the Kabanov model of markets with transaction costs and establish a link between this model and the general setting of Section \ref{section:general}.

We adopt the assumptions on the stochastic basis from Section \ref{section:general}.

Let $d \in \bbn$ be a number of assets and $S = (S_t)_{t \in [0, T]}$ be a $\bbr^d$-valued process associated with the prices of these assets. The price process $S$ is continuous, $\bbf$-adapted and all its entries are strictly positive. The first asset is a risk-free asset with a constant price $S^1 \equiv 1$ and, henceforth, will be referred to as the numeraire. By a change of measurement units, we set $S_0 = {\bf 1}$. 

The core idea of the Kabanov model \citep{kabanovHedgingLiquidationTransaction1999,kabanovMarketsTransactionCosts2010} is that the transaction fees can be represented by a closed convex proper cone $K$ such that $\bbr^d_+ \subset K$. Recall that a closed convex cone is proper if $K \cap (-K) = \{ 0 \}$. The dual cone $K^*$ is defined as $K^* = \{ y \in \bbr^d : \, xy \geq 0 \: \forall x \in K  \}$. The equivalent condition for a proper cone is ${\rm int} \, K^* \neq \varnothing$. Since $\bbr^d_+ \subset K$, the dual cone $K^* \setminus \{ 0 \} \subset {\rm int} \, \bbr^d_+$. Thanks to the bipolar theorem, $K^{**} = K$. 

We will also consider a solvency cone in terms of physical units $\wh K_t := K / S_t$. To be precise, we put $\varphi_t(x) := x / S_t$ and $\wh K_t := \varphi_t (K)$. Note that $\wh K^*_t = \varphi^{-1}_t (K^*)$. 
The closed convex proper cone $K$ generates a partial order $x \preccurlyeq y \Leftrightarrow y - x \in K$. 
\begin{definition}
   We say that a function $f: [0, T] \to \bbr^d$ is $K$-decreasing if $s \leq t$ implies $f(s) - f(t) \in K$.  
\end{definition}

\begin{definition}
    Any $\bbg$-adapted $\bbr^d$-valued $K$-decreasing c\`adl\`ag process $B$ is called a control or strategy.
\end{definition}

By definition, $B_{0-} = 0$ and $B_0$ is the charge of the initial point. Due to Lemma \ref{lemm:6.1}, $B$ is a process of finite variation. If $B = B^+ - B^-$ is the component-wise Jordan  decomposition, $B^+_t$ is the total monetary value of assets bought up to time $t$, whereas $B^-_t$ is the monetary value of asset sold up to time $t$. 

Due to Proposition I.3.13 of \citep{jacodLimitTheoremsStochastic2003}, process $B$ has a Radon--Nykodym derivative $\dot B$ with respect to the total variation process ${\rm Var} \, B := \sum_j {\rm Var} \, B^j$, where $\dot B$ is an optional process itself. By virtue of Lemma 3.6.1 of \citep{kabanovMarketsTransactionCosts2010}, this derivative can be chosen in such a way that $\dot B \in -K$ up to evanescence. Thus, our definition of strategy coincides with the one of \citep{kabanovMarketsTransactionCosts2010}, Section 3.6.

For a starting point $x \in K$, the controlled process 
\begin{equation*}
\wh V = \wh V^{x, B} := x + (1/S) \cdot B,
\end{equation*}
where $1/S_t = (1/S^1_t, \dots, 1/S^d_t)$ and $\cdot$ denotes the component-wise stochastic integration. Since $B$ is of bounded variation, the stochastic integral is simply the pathwise Lebesgue--Stieltjes integral. Recall that if $f$ is a continuous function and $g$ is of bounded variation then the Lebesgue--Stieltjes (L-S) and the Riemann--Stieltjes (R-S) integrations 
$$
{\rm (L-S)} \int_{[0,T]} f dg = f(0) g(0) + {\rm (R-S)} \int_0^T f dg. 
$$
The reason for this discrepancy is that the signed measure $dg$ has an atom at zero  $g(0)$. 

The entries of $\wh V_t$ can be interpreted as amounts of the assets in physical units at time $t$. Amounts of the assets in terms of the numeraire are $V := S \odot \wh V$, where $\odot$ is the Hadamard (component-wise) product. In other words, $\wh V = V / S$. 

\begin{definition}
 A strategy $B$ is admissible if $V^{x, B}_t \in K$ a.s. for all $t$.   
\end{definition}

Due to the right-continuity of $V^{x, B}$, this requirement is equivalent to $V^{x, B} \in K$ up to evanescence, or $ \wh V^{x, B} \in \wh K$ up to evanescence. 
 The set of all admissible strategies is denoted as $\cA(x)$. In terms of $\wh K$, that means that $\wh V^{x, B} \in \wh K$ a.s. 
It is easily seen that
 $\cA(x) \neq \varnothing$ for $x \in K$. Indeed, define the liquidation function $\ell$ of \citep{bouchardUtilityMaximizationReal2002}:
\beq
\ell(x) = \sup \{ y \in \bbr: \; x - y e_1 \in K   \},
\eeq
where $e_1$ is the first vector of the standard basic in $\bbr^d$. Note that $\ell$ is continuous. Put
\begin{equation*}
L^x_t := \ell(x)e_1 - x, \quad t \in [0, T]. 
\end{equation*}
It follows that $L^x$ is $K$-decreasing and deterministic, and 
\begin{equation*}
\wh V^{x, L^x}_t = \ell(x)e_1, \quad t \in [0, T]. 
\end{equation*}
If $x \in  K$, the liquidation function $\ell(x) \geq 0$, thus $L^x \in \cA(x)$.

The next theorem establishes that the Kabanov model is incorporated into the setup described in Section \ref{section:general}. It is proven in Section \ref{Section:proof}.

\begin{theorem}
\label{theo:embedding}
For a fixed $x \in K$, there exist a Borel function $\phi$ and a continuous function $\Phi_x$ such that 
\begin{equation*}
\cA_x := \{ B \in \mathfrak{B} \colon \Phi_x(Y, \phi(Y), B) \in \bbr^{\bbn}_+ \} = \cA(x).
\end{equation*}
\end{theorem}

In order for $\mathfrak{M}$ to be compact, we will impose a robust version of the no-arbitrage condition, cf. \citep{schachermayerFundamentalTheoremAsset2004,bayraktarExtendedWeakConvergence2020}. Loosely speaking, robustness means that arbitrage opportunities will not appear if the transaction fees undergo a slight decrease. We call a consistent price system (CPS) an $\bbr^d_+$-valued right-continuous $\bbf$-martingale $Z$ such that the quotient $ Z_t / S_t \in K^* \setminus \{ 0 \}$ a.s. for every $t \in [0, T]$. In other words, $Z \in \wh K^* \setminus \{ 0 \}$ up to evanescence. 
For a closed convex cone $G$, the $\varepsilon$-interior 
\beq
\varepsilon {\text -} {\rm int} \, G := \{ x \in \bbr^d : \; xy > \varepsilon |x| |y|  \; \; \forall y \in G^* \setminus \{ 0 \} \},
\eeq
where $|x|$ and $|y|$ are the standard $\ell^2$-norms of the vectors. A consistent price system $Z$ is an $\varepsilon$-consistent price system ($\varepsilon$-CPS) if $ Z_t / S_t \in \varepsilon {\text -} {\rm int}  K^* \setminus \{ 0 \}$ a.s. for every $t$. The set of all $\varepsilon$-consistent price systems is denoted as $\cM_\varepsilon (K)$.

\begin{assumption}
\label{A:4}
There exists $\varepsilon > 0$ such that $\cM_\varepsilon (K) \neq \varnothing$. 
\end{assumption}

The theorem below demonstrates that the no-arbitrage property is associated with compactness.
\begin{theorem}
\label{theo:apply}
Let Assumption \ref{A:4} hold and $x \in K$. Then Assumption \ref{A:3} is satisfied.
\end{theorem}

This theorem is proven in Section \ref{Section:proof}.

It remains to impose additional conditions that will guarantee Assumption \ref{A:1}. Since Assumption \ref{A:1} restricts goal functionals, the set of requirements depends on the goal. In the spirit of \citep{hePORTFOLIOCHOICEQUANTILES2011}, we provide several examples of goal functionals. 

Define $I: \cD_m \times \cC^d_{++} \times \cV \to \bbr^d$ as
\begin{equation*}
I(y, z, b) = z(T) \odot \int_{[0, T]} \frac{1}{z(s)} db(s).
\end{equation*}

\begin{lemma} \label{lemm:icont}
    The functional $I$ is continuous.
\end{lemma}
\begin{proof}
Fix a sequence $(y_n, z_n, b_n) \to (y, z, b)$. It suffices to verify that if from every subsequence $(y_{n_k}, z_{n_k}, b_{n_k})$ one can extract a further subsequence $(y_{n_{k_p}}, z_{n_{k_p}}, b_{n_{k_p}})$ such that $I((y_{n_{k_p}}, z_{n_{k_p}}, b_{n_{k_p}})) \to I (y, z, b)$. 

Fix an arbitrary subsequence $(y_{n_k}, z_{n_k}, b_{n_k})$. Convergence $b_{n_k} \to b$ in the Meyer--Zheng topology implies the convergence $b_{n_k} \to b$ in measure ${\rm Leb}([0, T)) + \delta_T$.
By the Riesz theorem, there is a subsequence $b_{n_{k_p}} \to b$ pointwise in a dense countable set $I \subset [0, T]$ with $T \in I$. By theorem 12.16 of \cite{protterFirstCourseReal1991}, we have $I((y_{n_{k_p}}, z_{n_{k_p}}, b_{n_{k_p}})) \to I (y, z, b)$. The lemma has been proven.  

\end{proof}

\begin{example}[Expected Utility Maximisation] \label{ex:eum}
\label{ex:utility}
\beq
\label{eq:utility}
    \sup_{B \in \cA(x)} \E U(V_T^{x, B}),
\eeq
where $x \in {\rm int} \, K$ and $U: \bbr^d \to \bbr \cup \{ \pm \infty \}$ is a utility function. We require that $U$ is upper-semi-continuous, but concavity is not imposed. 
\end{example}

\begin{assumption}
\label{A:5}
$U$ is upper-semi-continuous.
Additionally, for every $x \in K$, the family 
$$
\{  U^+( V^{x, B}_T, S_T): \, B \in \cA(x) \}
$$
is uniformly integrable.
\end{assumption}

\begin{theorem}
    Consider Example \ref{ex:eum}. Let Assumptions \ref{A:4} and \ref{A:5} be satisfied. Then the problem \eqref{eq:utility} attains a maximum.
\end{theorem}
\begin{proof}
    By Lemma \ref{lemm:icont} the mapping $I_{*}: \mathfrak{M} \to \cP(\bbr^d)$ defined as $I_*(\mu)(A) := \mu(I^{-1}(A))$ is continuous. The functional $J$ from \eqref{eq:utility} is defined as follows:
    \beq
    \label{eq:utility_measure}
    J(\mu) = \int_{\bbr} U(x) I_*(\mu)(dx).
    \eeq
    Due to Assumption \ref{A:5} and the Fatou lemma, functional $J$ satisfies Assumption \ref{A:1}. Due to Theorems \ref{theo:main} \ref{theo:embedding} and \ref{theo:apply}, the assertion follows.
\end{proof}

Verification of the uniform integrability condition in Assumption \ref{A:5}  might be cumbersome. We suggest a sufficient condition inspired by \citep{bayraktarContinuityUtilityMaximization2020,bayraktarExtendedWeakConvergence2020}.

\begin{lemma}
\label{lemm:ui}
Suppose that $U$ is upper-semi-continuous and there exist $C > 0$, $\gamma \in ]0, 1[$ and $q > 1 / (1 - \gamma)$ satisfying the following:

$(i)$ for all $(x, s) \in K \times \bbr^d_+$,
\beq
U(x, s) \leq C(1 + \ell^\gamma (x));
\eeq

$(ii)$ there exists a consistent price system $Z$ such that
\beq
\E \left( Z_T^{(1)} \right)^{1-q} < \infty.
\eeq
Then Assumption \ref{A:5} holds true.
\end{lemma}
The proof is given in Section \ref{Section:proof}.

\begin{example}[Goal Reaching]
\label{ex:goal_reaching}
\beq
\label{eq:goal_reaching}
\sup_{B \in \cA(x)} \E P(\ell(V_T^{x, B}) \geq b),
\eeq
where $b > 0$ is the goal intended to be reach by the terminal time. This formulation is motivated by a goal-reaching problem of \citep{browneReachingGoalsDeadline1999}.
This example is reduced to the previous case by putting $U(x) = {\bf 1}_{\ell(x) > b}$. Since this utility function is bounded and upper-semi-continuous, Assumption \ref{A:5} is fulfilled.
\end{example}

\begin{example}[Yaari's Dual Theory]
\label{ex:yaari}
\beq
\label{eq:yaari}
\sup_{B \in \cA(x)} \int_0^\infty w(P(\ell(V_T^{x, B}) > x)) dx,
\eeq
where $w: [0, 1] \to [0, 1]$ is a probability distortion, i.e. a continuous non-decreasing function with $w(0) = 0$ and $w(1) = 1$. The idea of distorted probabilty goes back to \citep{yaariDualTheoryChoice1987}. 

Contrary to Example \ref{ex:utility}, this functional requires the integration over the $\sigma$-finite Lebesgue measure. Because of this fact, we introduce

\begin{assumption}
\label{A:6}
There is a measurable function $g: \bbr_+ \to \bbr_+$ such that
\begin{equation*}
w \left(1 - F_{I_*(\mu)}(x) \right) \leq g(x) \quad {\rm a.e.}
\end{equation*}
for every $\mu \in \mathfrak{M}$ and
\begin{equation*}
\int_0^\infty g(x) dx < \infty.
\end{equation*}
\end{assumption}
\end{example}

\begin{theorem}
    Consider Example \ref{ex:yaari}. Let Assumptions \ref{A:4} and \ref{A:6} be satisfied. Then the problem \eqref{eq:yaari} attains a maximum.
\end{theorem}

\begin{proof}
The goal functional $J$ is given by
\beq
\label{eq:yaari_measure}
J(\mu) = \int_0^\infty w \left(1 - F_{I_*(\mu)}(x) \right) dx,
\eeq
where $F_\nu$ is a c.d.f. of a probability measure $\nu$ on $\bbr$.
Due to Assumption \ref{A:6} and the Fatou lemma, Assumption \ref{A:1} is fulfilled. Due to Theorems \ref{theo:main} \ref{theo:embedding} and \ref{theo:apply}, the assertion follows.
\end{proof}

\begin{example}[Cumulative Prospect Theory] \label{ex:cpt}
\label{ex:CPT}
Let $w_{\pm}: [0, 1] \to [0, 1]$ be continuous non-decreasing functions with $w(0) = 0$ and $w(1) = 1$. Henceforth, $w_{\pm}$ are called probability distortions. Let $U_{\pm}: \bbr^d \to \bbr_+$ with $U_\pm (0) = 0$ be continuous functions representing utilities. We require that $U_+$ is bounded. Let $W$ be an $\cF_T$-measurable $d$-dimensional r.v. For any r.v. $X \geq 0$, we define
\beq \label{eq:I_cpt}
I_{\pm}(X) := \int_0^\infty w_{\pm} \left(  P(X > x) \right) dx. 
\eeq
Note that $I_{-}(X)$ may be infinite.
For any $d$-dimensional r.v. $Z$, we put
\begin{equation*}
\cI_{\pm}(Z) := I_{\pm} ( U_{\pm} (Z - W)) 
\end{equation*}
and 
\begin{equation*}
\cI(Z) = \cI_+(Z) - \cI_-(Z).
\end{equation*}
We assume that  $U_+(x) = u_+(\ell(x)^+)$ and $U_-(x) = u_-(\ell(x)^-)$, where $u_\pm: \bbr \to \bbr$ are continuous increasing functions with $u(0) = 0$. Finally, the maximization problem itself is formulated as follows:
\beq
\label{eq:CPT_max}
\sup_{B \in \cA(x)} \cI(V^{x, B}_T).
\eeq

Following \citep{chauBehavioralInvestorsConic2020}, we assume that $U_+$ is bounded from above.
\end{example}

\begin{theorem}
    Consider Example \ref{ex:cpt}. Under Assumption \ref{A:4} and the conditions of the Example, the problem \eqref{eq:CPT_max} admits a solution.
\end{theorem}
\begin{proof}
    We set $\phi: \cD_m \to \cC^d \times \bbr^d$ in the following way: $\phi(Y) = (\phi^1(Y), \phi^2(Y)) = (S, W)$. This function exists thanks, again, to the Doob lemma. We put
    \begin{equation*}
    G_\pm(y, z, b) := U_\pm(I(y, z, b) - \phi^2(y))
    \end{equation*}
    and define $G_{\pm*} \colon \mathfrak{M} \to \cP(\bbr_+)$ as $G_{\pm*}(\mu)(A) := \mu(G^{-1}(A))$. Finally, the goal functional is
    \beq
    \label{eq:CPT_measure}
    J(\mu) = \int_0^\infty w_+ \left(1 - F_{G_+(\mu)} (x) \right) dx - \int_0^\infty w_- \left(1 - F_{G_-(\mu)} (x) \right) dx. 
    \eeq   
    Since $U_+$ is bounded, Assumption \ref{A:1} holds due to the Fatou lemma. Due to Theorems \ref{theo:main} \ref{theo:embedding} and \ref{theo:apply}, the assertion follows.
\end{proof}

\section{Behavioral storage management problem}
\label{Section:storage}

In this section, we study a storage management problem using the framework of Section \ref{section:general}. The main contribution of this section is the well-posedness of the storage management problem under cumulative prospect theory preferences.

We adopt the stochastic basis and notation from Section \ref{section:general}. 
Fix $k \in \bbn$.
Let $X$ be an $\bbf$-adapted $\bbr^k$-valued process with continuous paths (e.g. a diffusion process) with initial value $X_0 = x$. Process $X$ represents stochastic demand, whereas $x$ is the initial state of the uncontrolled storage. We also fix $K > 0$.

\begin{definition}
  A pair of $\bbr^d$-valued process $L := (L^+, L^-)$ is called an admissible strategy if they are $\bbg$-adapted, right-continuous, increasing, and the total variation
  \begin{equation} \label{eq:adm}
    \sum_{j=1}^d \left(L^{+,j}_T + L^{-,j}_T \right) \leq K.
\end{equation}
\end{definition}

Throughout this section, ``increasing'' is understood to mean ``nondecreasing''.

We denote $\cA$ the set of all admissible strategies.
Note that right-continuous increasing processes have left limits. The constraint \eqref{eq:adm} was inspired by \citep{harrison1983instantaneous,benevs1980some}. This reflects that the capacity to adjust the storage is inherently limited.

The storage evolves as
\begin{equation}
    Z_t := X_t + A_+ L^+_t - A_- L^-_t, 
\end{equation}
where $A_\pm = (a_{\pm,ij}) \in \bbr^{k \times d}$ are matrices of size $k \times d$. Note that $L^\pm_0$ is not necessarily zero, so an instantaneous jump at time $0$ is allowed: 
\[
Z_0 = x + A_+ L^+_0 - A_- L^-_0.
\]

The running cost of maintenance of this storage is given by a functional
\begin{equation}
   W = W^L := \int_0^T g(s, Z_s) ds + \int_{[0,T]} h_+(s)  d L^+_s + \int_{[0,T]} h_-(s) d L^-_s + G(Z_T),
\end{equation}
where $h_\pm \colon [0, T] \to \bbr_+^d$ are continuous functions and $g \colon [0, T] \times \bbr^k \to \bbr_+$ and $G \colon \bbr^k \to \bbr_+$ are bounded continuous functions. The notation $\int_0^T h_\pm(s)  d L^\pm_s$ means
\[
 \int_{[0,T]} h_\pm(s)  d L^\pm_s = \sum_{j=1}^d \int_{[0,T]} h_\pm^j(s) dL^{\pm,j}_s.
\]

In particular, if $d=k=1$ and $A_\pm = 1$, this reduces to the one-dimensional model studied by \citep{harrison1983instantaneous}:
\begin{equation*}
    Z_t := X_t +  L^+_t - L^-_t.
\end{equation*}

Note that by the Doob lemma, $X = \phi(Y)$ for some Borel function $\phi: \cD_m \to \cC_k$, and $Y$ was defined in Section \ref{section:general}. 

Let 
\[
\mathfrak{M} := \left\{ \cL_\P(Y, \phi(Y), L) \colon L \in \cA \right\}
\]
be a family of probability measures on the product space $\cD_m \times \cC_k \times \cV^2$.

\begin{lemma}  \label{lemm:compact_storage}
    The family of probability measures $\mathfrak{M}$ is compact in the space $\cP(\cD_m \times \cC_k \times \cV^2)$ and, as such, Assumption \ref{A:3} holds.
\end{lemma}
\begin{proof}
    Let $V_K:= \{ (x_1, x_2) \in \cV^2 \colon {\rm Var} \, x_1 + {\rm Var} \, x_2 \leq K \}$. By the Helly theorem, $V_K$ is compact in the Meyer--Zheng topology. Since $L \in V_K$ with probability one for all $L \in \cA$, the family $\mathfrak{M}$ is tight, and, thanks to the Prokhorov theorem, compact. 
\end{proof}

\begin{lemma} \label{lemm:adm_storage}
    The set $\cA$ satisfies Assumption \ref{A:2}. 
\end{lemma}
\begin{proof}
    Convergence of functions from $\cV$ in the Meyer--Zheng topology implies weak convergence. Choose a countable family of positive bounded continuous functions $(f_i)_{i \in \bbn} \subset \cC_{2d}$ defining weak convergence in $\cV^2$.
    Then we set 
    \[
    \Phi_{i}(l) := \int_{[0,T]} f(t) dl(t).  
    \]
    Note that $l$ is increasing if and only if $\Phi_{i} (l) \in \bbr_+$ for all $i$.
    Also, we put
    \[
    \Phi_\infty(l^+, l^-) := K - l^+_T - l^-_T = K - \int_{[0,T]} (dl^+(t) + dl^-(t)).
    \]
    $l$ satisfies \eqref{eq:adm}.
\end{proof}

We define a mapping $\psi \colon \cC_k \times \cV^2 \to \cD_k$ 
\begin{equation}
    \psi(y, l^+, l^-) :=  y + A^+ l^+ - A^- l^-.
\end{equation}
Define $\Theta \colon \cC_k \times \cV^2  \to \bbr_+$ by
\begin{multline}
    \Theta(y, l^+, l^-) := \int_0^T g(t,\psi(y,l^+,l^-)(t)) dt + \int_0^T h_+(t)  d l^+_t + \\ \int_0^T h_-(t)  d l^-_t + G(\psi(y,l^+,l^-)(T)).
\end{multline}
It follows directly from the boundedness of $h_\pm$, $g$ and $G$ that  $\Theta$ is bounded on $\cC_k \times \cH_K$, where 
$$
\cH_K = \{(f_1,f_2) \in \cV^2 \colon \sum_{j=1}^d \left( {\rm Var}_T \, f^j_1 + {\rm Var}_T \, f^j_2 \right) \leq K \}.
$$
\begin{lemma} \label{lemm:icont_storage}
    The function $\Theta$ is continuous.
\end{lemma}

This lemma is similar to Lemma \ref{lemm:icont}, so the detailed proof is omitted.

Define a pushforward mapping  $\Theta^* \colon \cP(\cD_m \times \cC_k \times \cV^2) \to \cP(\bbr_+)$ by $\Theta^*(\mu)(A) := \mu(\cD^k_T \times \Theta^{-1}(A))$. From the previous lemma, $\Theta^*$ is continuous on $\mathfrak{M}$.

Similarly to Section \ref{section:kabanov_setting}, we investigate well-posedness of different optimization problems. Recall that $W^L$ represents the expenses, so the economic agent would minimize them. 

\begin{example}[Expected Expenditures Minimization] \label{ex:eem}
    Consider a problem
    \begin{equation}
        \inf_{L \in \cA} \E W^L = -\sup_{L \in \cA} \E (-W^L). 
    \end{equation}

    where $\mu = \cL_\P(Y, \phi(Y), L)$ and $\mu \in \mathfrak{M}$. Since $\Theta^*$ is continuous on $\mathfrak{M}$,  Assumption \ref{A:1} holds. 
\end{example}

\begin{theorem}
    The problem of Example \ref{ex:eem} admits an optimal solution. 
\end{theorem}
\begin{proof}
    Note that 
    \[
    E W^L = J(\mu) := \int_{\bbr_+} x \Theta^*(\mu)(dx).
    \]
    Since $\Theta \leq M$ for some $M > 0$,
    \[
    J(\mu) = \int_{\bbr_+} U(x) \Theta^*(\mu)(dx), \quad U(x) = \min(x,M).
    \]
    Since $U$ is bounded, Assumption \ref{A:1} holds by the Fatou lemma. Thanks to Lemmas \ref{lemm:icont_storage} and \ref{lemm:adm_storage} and Theorem \ref{theo:main}, the assertion follows.
\end{proof}

\begin{example}[Goal Reaching]
    Consider a problem
    \begin{equation}
        \inf_{L \in \cA}  \P( W^L > b). 
    \end{equation}
    It follows that 
    \[
    \inf_{L \in \cA} \P( W^L > b) = 1 - \sup_{L \in \cA} \E U(W^L),
    \]
    where $U(x) = {\bf 1}_{x \leq b}$. This function is bounded and upper-semi-continuous, so this situation is similar to the previous example. 
\end{example}

The above examples illustrate the versatility of our methodology. We now turn to an example motivated by cumulative prospect theory.

\begin{example}[Cumulative Prospect Theory] \label{ex:cpt_s}
    Similarly to Example \ref{ex:cpt}, we define functions $I_\pm$ by \eqref{eq:I_cpt}. Then, we fix a $\cF_T$-measurable r.v. $\zeta \geq 0$. By the Doob lemma, $\zeta = \phi_2(Y)$ for some Borel function $\phi_2$. This variable represents a benchmark level for the agent. Let $U_\pm \colon \bbr_+ \to \bbr_+$ be continuous increasing functions, $U_\pm(0) = 0$. We assume that $U_+$ is bounded.
    We define 
    \begin{equation}
        \cI_\pm(W) := I_\pm(U_\pm(\zeta - W)^\pm), \quad \cI(W) := \cI_+(W) - \cI_-(W),
    \end{equation}
    where $(x)^+$ and $(x)^-$ are positive and negative parts of $x$ respectively.
    Then the optimization problem reads
    \begin{equation} \label{eq:cpt_s}
        \sup_{L \in \cA}  \cI(W^L).
    \end{equation}
\end{example}

\begin{theorem}
    Under the conditions of Example \ref{ex:cpt_s}, the problem \eqref{eq:cpt_s} admits an optimal solution.
\end{theorem}
\begin{proof}
    We put
    \[
    G_\pm(y, l^+, l^-) := U_\pm(\Theta(\phi(y), l^+, l^-) - \phi_2(y))
    \]
    and define $G_{\pm*} \colon \mathfrak{M} \to \cP(\bbr_+)$ as $G_{\pm*}(\mu)(A):=\mu(G^{-1}_\pm(A))$. The goal functional can be written as
    \begin{equation} \label{eq:cpt_s}
          J(\mu) := \int_0^\infty w_+(1 - F_{G_{+*}(\mu)}(x))dx - \int_0^\infty w_-(1 - F_{G_{-*}(\mu)}(x))dx,
    \end{equation} 
    where $F_\mu$ is the cumulative distribution function of probability measure $\mu \in \cP(\bbr_+)$.
    By boundedness of $U_+$ and the Fatou lemma, $J$ satisfies \ref{A:1}. Thanks to Lemmas \ref{lemm:icont_storage} and \ref{lemm:adm_storage} and Theorem \ref{theo:main}, the result follows.
\end{proof}

\section{Proofs}
\label{Section:proof}

\subsection{Proof of Theorem \ref{theo:general}}
\label{section:proof_general}
 We put ${\rm Var}_T B := \sum_{j=1}^d  {\rm Var}_T B^j$ the total variation of $B \in \mathfrak{B}$. Let $(\mu_n)_{n=1}^\infty \subseteq \mathfrak{M}$ be an arbitrary sequence. By definition, we have $B^n \in \cA$ such that $\cL_P (Y, \phi(Y), B^n ) = \mu_n$.

By virtue of Assumption \ref{A:3} and the Helly theorem for $\cV$, the family of r.v. $((Y, \phi(Y), B^n))_{n=1}^\infty$ is tight and, thanks to the Prokhorov theorem, relatively compact. Thus, there is a cluster point $\mu$ and, up to a subsequence, $\cL_P(Y, \phi(Y), B^n) \to \mu$. With an abuse of notation, the additional subscripts for subsequences are omitted.

By the Skorokhod representation theorem, there are another probability space $(\tilde \Omega, \tilde \cF, \tilde P)$ and r.v. $((\tilde Y^n, \tilde \zeta^n, \tilde B^n))_{n=1}^\infty$ and $(\tilde Y^\dagger, \tilde \zeta^\dagger, \tilde B^\dagger)$ such that the distributions of $\cL_{ \tilde P} (\tilde Y^n, \tilde \zeta^n, \tilde B^n) = \cL_P (Y, \phi(Y), B^n)$ for all $n$, $\cL_{\tilde P} (\tilde Y^\dagger, \tilde \zeta^\dagger, \tilde B^\dagger ) = \mu$ and the sequence $(\tilde Y^n, \tilde \zeta^n, \tilde B^n) \to (\tilde Y^\dagger, \tilde \zeta^\dagger, \tilde B^\dagger)$ a.s. in $\cD_m \times \cE \times \cV$.

After that, thanks to the upper semi-continuity of $\Phi$, we have 
\begin{equation*}
 \limsup_{n \to \infty} \Phi(\tilde Y^n, \tilde \zeta^n, \tilde B^n) \leq \Phi(\tilde Y^\dagger, \tilde \zeta^\dagger, \tilde B^\dagger), \quad {\rm a.s.}   
\end{equation*}
(the comparison is componentwise) and, as such, $\Phi(\tilde Y^\dagger, \tilde \zeta^\dagger, \tilde B^\dagger) \in \bbr^\bbn_+$ a.s. 

Now we "resettle" process $\tilde B^\dagger$ to the original probability space. Consider a regular conditional distribution $Q: \cD_m \times \cB(\cE \times \cD_{d, MZ}) \to [0, 1]$ of $\tilde \zeta^\dagger \times \tilde B^\dagger$ with respect to $\tilde Y^\dagger$. By Lemma 4.22 of \citep{kallenbergFoundationsModernProbability2021}, we have $Q(y, \cdot) = \cL_P(h(y, \xi))$ for some measurable function $h$ and all $y \in \cD_m$. We define $(\zeta^\dagger, B^\dagger) := h(Y, \xi)$. It follows that the laws $\cL_P (Y, \zeta^\dagger, B^\dagger ) = \cL_{\tilde P}(\tilde Y^\dagger, \tilde \zeta^\dagger, \tilde B^\dagger)$.

Note that $\cL_P (\zeta^\dagger , \xi) = \cL_P (\psi(Y), \xi)$ and  $\cL_P(\psi(Y), \xi) = \cL_P(\psi(Y)) \otimes \cL_P(\xi)$ implying that $\zeta^\dagger$ and $\xi$ are independent. Thanks to Lemma 29 of \citep{chauSkorohodRepresentationTheorem2017}, $\zeta^\dagger$ is $\sigma\{Y\}$-measurable. Hence, $\zeta^\dagger = \phi^\dagger(Y)$ a.s. for some Borel function $\phi$. It can be easily established that $\phi^\dagger(Y) = \phi(Y)$ a.s. 

It remains to verify that $B^\dagger$ is $\bbg$-adapted. Following the lines of \citep{chauSkorohodRepresentationTheorem2017}, we define $_t \! Y_s := (Y_s - Y_t) {\bf 1}_{[t, \infty)}(s)$ and $Y^t_s := Y_{s \wedge t}$. Then $B^\dagger_t$ is independent of $_t \! Y$. Indeed, let $t \leq u_1 < \dots < u_k \leq T$. Then
\begin{multline*}
\cL_{\tilde P}(\tilde B^n_t, \tilde Y^n_{u_1} - \tilde Y^n_t, \dots, \tilde Y^n_{u_k} - \tilde Y^n_t ) = \cL_P ( B^n_t,  Y_{u_1} - Y_t, \dots,  Y_{u_k} -  Y_t ) = \\
\cL_P(B^n_t) \otimes \cL_P ( Y_{u_1} - Y_t, \dots,  Y_{u_k} -  Y_t ) = \cL_{\tilde P}(\tilde B^n_t) \otimes \cL_{\tilde P}( \tilde Y^n_{u_1} - \tilde Y^n_t, \dots, \tilde Y^n_{u_k} - \tilde Y^n_t).
\end{multline*}
By approaching $n \to \infty$, we obtain
\begin{equation*}
\cL_{\tilde P}(\tilde B^\dagger_t, \tilde Y^\dagger_{u_1} - \tilde Y^\dagger_t, \dots, \tilde Y^\dagger_{u_k} - \tilde Y^\dagger_t  ) = 
\cL_{ \tilde P}(\tilde B^*_t ) \otimes \cL_{\tilde P}( \tilde Y^\dagger_{u_1} - \tilde Y^\dagger_t, \dots, \tilde Y^\dagger_{u_k} - \tilde Y^\dagger_t )
\end{equation*}
and, as such,
\begin{equation*}
\cL_P (B^\dagger_t,  Y_{u_1} -  Y_t, \dots,  Y_{u_k} -  Y_t ) = 
\cL_P ( B^\dagger_t ) \otimes \cL_P (  Y_{u_1} -  Y_t, \dots,  Y_{u_k} -  Y_t)
\end{equation*}
By construction, $B^\dagger_t$ is $\sigma \{ \,_t \! Y, Y^t, \xi \}$-measurable. Again, due to Lemma 29 of \citep{chauSkorohodRepresentationTheorem2017}, it follows that  $B^\dagger_t$ is $\sigma \{ Y^t, \xi \}$-measurable. Thus $B^\dagger$ is $\bbg$-adapted. Hence, $B^\dagger \in \cA$, and this concludes the proof.

\subsection{Proof of Theorem \ref{theo:embedding}}
Let $\cC_{++}$ a family of strictly positive functions from $\cC_1$.

At first, recall that $S$ is $\bbf$-adapted and $\bbf$ is generated by $Y$, hence, due to the Doob lemma, $S = \phi(Y)$, where $\phi: \cD_m \to \cC^d_{++}$ is a Borel function. Now it is required to transplant the restrictions into the language of the general framework. To this end, we select dense countable sets $\{ a_k \}_{k=1}^\infty \subset -K^*$ and $I \subset [0, T]$ with $T \in I$. Also, we select a dense countable set $\{f_l \}_{l=1}^\infty \subset C([0, T], \bbr_+)$. For $k, \, l \in \bbn$ with $t_i \leq t_j$, put 
\begin{equation*}
\Phi_{kl}(y, z, b) := a_k \int_{[0,T]} f_l d b.
\end{equation*}
The functionals $\Phi_{kl}: \cD_m \times \cC^d_{++} \times \cV \to \bbr$ are obviously continuous. Clearly, $B$ is $K$-decreasing iff $\Phi_{kl} (Y, S, B) \geq 0$ for all indices. Indeed, a signed measure $\mu$ on $[0, T]$ is positive iff $\mu(f_l) \geq 0$ for all $l$.

Now let us formalise the admissibility condition. Fix $x \in K$ and define 
\begin{equation*}
\Phi'_{jk}(y, z, b) := a_k \left( \left( x +  \int_{[0, t_j]} \frac{1}{z(s)} db(s) \right) \odot z \right),
\end{equation*}
where $1/z = (1/z^1, \dots, 1/z^d)$ is a componentwise quotient. Immediately, $\wh V \in \wh K$ iff $\Phi'_{jk}(Y, S, B) \geq 0$ for all indices.
By virtue of Theorem 12.16 of \citep{protterFirstCourseReal1991}, the functional $\Phi'_{jk}(y, z, b): \cD_m \times \cC^d_{++} \times \cV \to \bbr$ is continuous. Finally, define
\begin{equation*}
\Phi(y, z, b) := ( \Phi_{kl}(y, z, b), \, \Phi'_{jk}(y, z, b)  )_{jkl \in \bbn}.
\end{equation*}
It is easily seen that $\Phi$ satisfies Assumption \ref{A:2}.
With the above notation, an adapted right-continuous process of finite variation $B$ is an admissible strategy if and only if the value $\Phi(Y, \phi(Y), B) \in \bbr^\bbn_+$ a.s.

\subsection{Proof of Theorem \ref{theo:apply}}

In order to establish tightness, we show an estimate based on the $\varepsilon$-CPS. Put $\Lambda := K^* \cap \{ z \in \bbr^d: \, z^{(1)} = 1  \}$. Note that $\Lambda$ is compact as $K^* \setminus \{ 0 \} \subset {\rm int} \, \bbr^d_+$. Since $\cF_0$ is a null sigma algebra, $Z_0$ is a constant. Without loss of generality, we may assume that $Z^{(1)}_0 = 1$.
We introduce a function:
\begin{equation}
\wp (x) = \inf \{ y \in \bbr: \: ye_1 - x \in K \}.
\end{equation}
Henceforth, we will refer to it as purchase function. Standard arguments (see, e.g., \citep{bouchardUtilityMaximizationReal2002}) imply the following dual representation:
\begin{equation}
\wp (x) = \sup_{y \in \Lambda} xy. 
\end{equation}
Note that the liquidation function has a similar representation (see, again, \citep{bouchardUtilityMaximizationReal2002})
\begin{equation}
\ell (x) = \inf_{y \in \Lambda} xy. 
\end{equation}
Immediately, $\wp (x) = - \ell (-x)$.

\begin{lemma}
\label{lemma:4.2}
Let Assumption \ref{A:2} hold and let $Z$ be an $\varepsilon$-CPS with $Z^{(1)}_0 = 1$. Let \\$Q = Z^{(1)}_T P$ be an equivalent probability measure. Finally, fix $x \in {\rm int} \, K$ and let \\$B \in \cA(x)$. Then 
\beq
\E_Q {\rm Var}_T B \leq \frac{1}{\varepsilon} \wp (x).
\eeq

\end{lemma}

\begin{proof}
Without loss of generality, $Z^{(1)}_0 = 1$. Put $M_t := Z_t / Z^{(1)}_t$.
Due to Lemma \ref{lemm:6.1}, $\dot B_t \in -K$. By definition of $\varepsilon$-CPS,
\begin{equation*}
(\dot B_t / S_t) M_t  = \dot B_t (M_t / S_t) < -\varepsilon |\dot B_t| |M_t / S_t|.
\end{equation*}
Clearly, $|\dot B_t(\omega)|$ = 1 outside a set of measure $d {\rm Var} \, B_T(\omega) dP(\omega)$ zero. By changing the derivative on a set of $d {\rm Var} \, B_T(\omega) dP(\omega)$-measure zero, we arrive at $|\dot B| \equiv 1$. Also, note that $M^{(1)}_t \equiv 1$ and $S^{(1)}_t \equiv 1$ implying that $|M_t / S_t| \geq 1$. Finally, we obtain  $\dot B_t (M_t / S_t ) \leq -\varepsilon$.

Since $M$ is a $Q$-martingale, Lemma \ref{lemma:4.1} yields that $M \wh V$ is a $Q$-supermartingale.
Due to $\wh V_{0-} = x$ and $S_0 = \1$,
\begin{equation*}
\E_Q M \cdot \wh V_T  = 
x \E_Q M_0 + \E ((M/S) \dot B) \cdot {\rm Var} \, B_T \leq x \E_Q M_0 - \varepsilon \E_Q {\rm Var} \, B_T.
\end{equation*}
Finally,
\begin{equation*}
\varepsilon \E_Q  {\rm Var} \, B_T \leq x \E_Q M_0 - \E_Q M_T \wh V_T \leq x \E_Q M_0.
\end{equation*}
Since $M_0 \in \Lambda$ a.s., the assertion follows.
\end{proof}

\smallskip
\noindent
{\sl Completion of the proof.}
Due to Lemma \ref{lemma:4.2} and the Markov inequality, a family of measures
\begin{equation*}
\{  \cL_Q ( {\rm Var}_T \, B ): \; B \in \cA(x)  \}
\end{equation*}
is tight. Since $P \ll Q$, 
\begin{equation*}
\{  \cL_P ( {\rm Var}_T \, B ): \; B \in \cA(x)  \}
\end{equation*}
is tight as well. 
As $\{f : \; {\rm Var}_T \, f \leq c \}$ are compact sets in $D_{MZ}([0, T], \bbr^d)$, the assertion follows.

\subsection{Proof of Lemma \ref{lemm:ui}}

Put $p := q / (q - 1)$. Clearly, $1/p > \gamma$. It suffices to establish that, for a fixed $x \in {\rm int}\, K$,
\begin{equation*}
\sup_{B \in \cA(x)} \E \left(V^{x, B}_T \right)^{1/p} < \infty.
\end{equation*}
Put $Q := Z^{(1)}_T P$ and $M := Z / Z^{(1)}$. By Lemma \ref{lemma:4.1}, $\wh V^{x, B} Z$ is a $P$-supermartingale implying that $\wh V^{x, B} M$ is a $Q$-supermartingale. By definition, $V^{x, B}_T - \ell(V^{x, B}_T) e_1  \in K$, where $e_1$ is the first vector of the standard basis in $\bbr^d$. As $Z / S \in K^*$ and $M^{(1)} \equiv 1$, we have
\begin{equation*}
\ell(V^{x, B}_T) = (M_T / S_T) \ell(V^{x, B}_T) e_1 \leq (M_T / S_T) V^{x, B}_T = M_T \wh V^{x, B}_T.
\end{equation*}
Then by the H\"{o}lder inequality
\begin{multline*}
\E \left( \ell(V^{x, B}_T) \right)^{1/p} \leq \E \left(M_T \wh V^{x, B}_T \right)^{1/p} = \E_Q \left[ \left(M_T \wh V^{x, B}_T \right)^{1/p} \left(Z^{(1)}_T \right)^{-1} \right] \leq \\
\left(\E_Q \left[ M_T \wh V^{x, B}_T \right] \right)^{1/p} \left( \E_Q \left(Z^{(1)}_T \right)^{-q} \right)^{1/q} \leq x^{1/p}  \left( \E \left(Z^{(1)}_T \right)^{1-q} \right)^{1/q} < \infty.
\end{multline*}
That concludes the proof.

\backmatter









\begin{appendices}

\section{Appendix}\label{Section:appendix}

\subsection{$K$-increasing (decreasing) functions are of bounded variation}

A right-continuous function of bounded variation $f: [0, T] \to \bbr$ induces a signed measure $\mu$ on $[0, T]$ with $\mu([0, t]) = f(t)$. In particular, $f(0)$ is $\mu(\{0\})$. If the Hahn decomposition $f = f_1 - f_2$ with $f_2(0) = 0$, then ${\rm Var} \, f = f_1 + f_2$. If a function  $f: [0, T] \to \bbr^d$ with entries $f^i$ being right-continuous functions of bounded variation,  set ${\rm Var} \, f : = \sum_i {\rm Var} \, f^i$.

\begin{lemma}
\label{lemm:6.1}
Let $f: \, [0, T] \to \bbr^d$ be a right-continuous function and $K$ be a closed proper convex cone. Then the following are equivalent:

$(i)$ $f$ is $K$-increasing ($K$-decreasing);

$(ii)$ $f$ is of bounded variation with the Radon--Nikodym derivative $\dot f := d f / d {\rm Var} f$ evolving in $K$ ($-K$), where ${\rm Var} \, f(t)$ is the total variation of $f$ on $[0, t]$.
\end{lemma}
\begin{proof}
$(ii) \Rightarrow (i)$ is apparent. It remains to verify $(i) \Rightarrow (ii)$. Consider the case of $f$ being $K$-increasing. As $K$ is proper, $K^*$ contains a basis $a_1, \dots, a_d$ of $\bbr^d$. Denote scalar products $g_i(t) = a_i f(t)$. As $f(t) - f(s) \in K$, we have that $g_i(t) \geq g_i(s)$, i.e. $g_i$ are increasing.  The coefficients $c_i(t)$ of a linear combination
\begin{equation*}
f(t) := \sum_{i=1}^d a_i c_i(t) 
\end{equation*}
can be obtained as linear combinations of scalar products $g_i(t)$:
\begin{equation*}
c(t) = G^{-1} g(t),
\end{equation*}
where $G$ is the Gram matrix of $\{ a_i \}_{i=1}^d$. It implies that $c_i(t)$ are right-continuous of bounded variation and so is $f$. In this case, ${\rm Var} \, f$ is well-defined. Since 
\begin{equation*}
a_i(f(t) - f(s)) = \int_{]s, t]} a_i \dot f(u) d {\rm Var} f(u) \geq 0,
\end{equation*}
$\dot f \in K$ up to a set of $d {\rm Var} \, f$-measure zero. By redefining $\dot f$ on this set, we arrive at the required statement. 
\end{proof}

\subsection{An auxiliary estimate}

The statement below is similar to Lemma 3.6.2 of \citep{kabanovMarketsTransactionCosts2010}. Let $\cM$ be a family of right-continuous $\bbf$-martingales $Z$ with $Z_t \in \wh K_t$ a.s. for every $t \in [0, T]$ (and, as such, up to evanescence). Denote $\cX$ the set of all $d$-dimensional right-continuous processes of bounded variation $X$ with the Radon--Nikodym derivative (with respect to the total variation process ${\rm Var} \, X$) $\dot X \in - \wh K$ and $X + a {\bf 1} \in \wh K$  up to evanescence for some $a \in \bbr$. By convention, $X_{0-} = 0$.

\begin{lemma}
\label{lemma:4.1}
    Let $X \in \cX$ and $Z \in \cM$. Then the scalar product $ZM$ is a supermartingale, and 
    \beq
    \E (Z \dot X) \cdot {\rm Var}_T \, X \geq  \E Z_T X_T.
    \eeq
\end{lemma}
\begin{proof}
    By the product formula,
    \begin{equation*}
    Z  X = X_- \cdot Z + Z \cdot X_T.
    \end{equation*}
    The first term is a local martingale. As $Z_t \dot X_t \leq 0$, the latter term is a negative decreasing process.  For every entry, we have
    \begin{equation*}
    X^j_- \cdot Z^j = Z^jX^j - Z^j \dot X^j \cdot {\rm Var} \, X  \geq Z^j X^j \geq -aZ^j
    \end{equation*}
    for some $a \in \bbr_+$ by definition of $\cX$. Recall that a local martingale bounded from below by a martingale is a supermartingale. The terminal value of the decreasing process $Z \cdot X_T = Z_T X_T - X_- \cdot Z_T $. Since $X_- \cdot Z$ is a supermartingale, we obtain that $\E X_- \cdot Z_T \leq \E X_- \cdot Z_0 = 0$. Finally, $Z_T X_T \geq -a Z_T \1$ implying that $(Z \dot X) \cdot {\rm Var } \, X_T$ is bounded from below by an integrable random variable. Thus, $(Z \dot X) \cdot {\rm Var } \, X$ is a supermartingale and so is $ZX$. Immediately,
    \begin{equation*}
    \E (Z \dot X) \cdot {\rm Var}_T \, X = \E(Z_T X_T - X_- \cdot Z_T ) \geq \E Z_T X_T.
    \end{equation*}
    That concludes the proof.
\end{proof}

\section*{Declarations}

\subsection*{Ethical Approval}
Not applicable.

\subsection*{Competing interests}
The authors have no relevant financial or non-financial interests to disclose.

\subsection*{Authors' contributions}

All authors wrote and reviewed the manuscript.





\end{appendices}


\bibliography{sn-bibliography}

\end{document}